%% file: cure-cross-sg.tex
\pdfoutput=1
\documentclass[conference]{IEEEtran}
\input{preamble}
\input{newcommands}

\begin{document}
\title{Tunable Causal Consistency: Specification and Implementation}



 \author{\IEEEauthorblockN{Xue Jiang, Hengfeng Wei, Yu Huang}
 \IEEEauthorblockA{\textit{State Key Laboratory for Novel Software Technology,
   Nanjing University, China} \\
   xuejiang1225@gmail.com, \{hfwei, yuhuang\}@nju.edu.cn}%
 }

%
%
\maketitle
\begin{abstract}
  \input{sections/abstract}
\end{abstract}

\begin{IEEEkeywords}
  Causal Consistency; Tunable Consistency; Session Guarantees
\end{IEEEkeywords}

\input{sections/intro}
\input{sections/tcc}
\input{sections/curelc}
\input{sections/evaluation}

\input{sections/related-work}
\input{sections/conclusion}
\bibliographystyle{IEEEtran}
\bibliography{cure}

\clearpage
\appendices
\input{sections/appendix-protocol}

\input{sections/appendix-proof}
\end{document}

%% file: preamble.tex
\usepackage{xcolor} 
\usepackage{paralist}
\usepackage{cite}
\usepackage{comment}

\usepackage{hyperref}

\usepackage{amsmath}
\usepackage{amssymb, amsthm, mathtools}
\usepackage{centernot}

\theoremstyle{definition}
\newtheorem{definition}{Definition}
\newtheorem{theorem}{Theorem}

\newtheorem{example}{Example}

\usepackage{fancyhdr}
\usepackage{array}   
\newcolumntype{C}{>{$}c<{$}} 

\usepackage{multirow}

\usepackage{graphicx}
\usepackage{subfigure}
\usepackage{float}

\usepackage{algorithm}
\usepackage[noend]{algpseudocode} 
\newcommand{\hStatex}[0]{\vspace{4pt}}

\usepackage[textwidth = 0.30\textwidth, textsize = footnotesize]{todonotes}

%% file: newcommands.tex
\newcommand{\set}[1]{\{#1\}}

\newcommand{\comp}{\;;}

\newcommand{\teal}[1]{\textcolor{teal}{#1}}
\newcommand{\purple}[1]{\textcolor{purple}{#1}}

\newcommand{\op}{\texttt{op}}
\newcommand{\rval}{\texttt{rval}}

\newcommand{\ctxt}{\texttt{ctxt}}

\newcommand{\eval}{\texttt{eval}}

\newcommand{\Val}{\textsl{Val}}

\newcommand{\Op}{\textsl{Op}}

\newcommand{\Level}{L}

\newcommand{\intreg}{\textsf{reg}}

\newcommand{\rd}{\texttt{rd}}
\renewcommand{\wr}{\texttt{wr}}

\newcommand{\cvc}{\texttt{cvc}}

\newcommand{\ok}{\texttt{ok}}
\newcommand{\pvc}{\texttt{pvc}}
\newcommand{\css}{\texttt{gsvc}}
\newcommand{\clock}{\texttt{clock}}
\newcommand{\updates}{\texttt{updates}}
\newcommand{\store}{\texttt{store}}
\newcommand{\pmc}{\texttt{PMC}}

\newcommand{\readkey}{\textsc{get}}
\newcommand{\updatekey}{\textsc{put}}
\newcommand{\readrequest}{\textsc{get-req}}
\newcommand{\updaterequest}{\textsc{put-req}}
\newcommand{\propagate}{\textsc{propagate-updates}}
\newcommand{\heartbeat}{\textsc{heartbeat}}
\newcommand{\replicate}{\textsc{replicate}}
\newcommand{\bcast}{\textsc{bcast-pvc}}
\newcommand{\updatecss}{\textsc{update-css}}

\newcommand{\send}{{\bf send}\;}

\newcommand{\wait}{{\bf wait}\;}
\newcommand{\until}{{\bf until}\;}

\newcommand{\partition}{\textsc{partition}}
\newcommand{\tcc}{TCC}
\newcommand{\level}{\texttt{lvl}}

\newcommand{\hwvc}{\texttt{hwvc}}
\newcommand{\hrvc}{\texttt{hrvc}}

\newcommand{\pput}{\textsc{put}}
\newcommand{\get}{\textsc{get}}

\newcommand{\tccstore}{\textsc{TCCStore}}

\newcommand{\mr}{\texttt{mr}}
\newcommand{\ryw}{\texttt{ryw}}

\newcommand{\wfr}{\texttt{wfr}}

\newcommand{\mw}{\texttt{mw}}

\newcommand{\rel}[1]{\xrightarrow{#1}}

\newcommand{\so}{\texttt{so}}

\newcommand{\vis}{\texttt{vis}}
\newcommand{\ar}{\texttt{ar}}

\newcommand{\creads}{\textsl{ClientReads}}
\newcommand{\cwrites}{\textsl{ClientWrites}}
\newcommand{\swrites}{\textsl{StableWrites}}

\newcommand{\retval}{\textsc{RVal}}

\newcommand{\ec}{\texttt{ec}}
\newcommand{\cc}{\texttt{cc}}


%% file: sections/abstract.tex
To achieve high availability and
low latency, distributed data stores often geographically replicate data
at multiple sites called replicas.
However, this introduces the data consistency problem.
Due to the fundamental tradeoffs among consistency, availability,
and latency in the presence of network partition,
no a one-size-fits-all consistency model exists.

To meet the needs of different applications,
many popular data stores provide tunable consistency,
allowing clients to specify the consistency level per individual operation.
In this paper, we propose tunable causal consistency (\tcc).
It allows clients to choose the desired session guarantee for each operation,
from the well-known four session guarantees,
i.e., read your writes, monotonic reads, monotonic writes, and writes follow reads.
Specifically, we first propose a formal specification of \tcc{}
in an extended $(vis, ar)$ framework originally proposed by Burckhardt et al.
Then we design a \tcc{} protocol
and develop a prototype distributed key-value store called \tccstore{}.
We evaluate \tccstore{} on Aliyun.
The latency is less than 38ms for all workloads and
the throughput is up to about 2800 operations per second.
We also show that \tcc{} achieves better performance than causal consistency
and requires a negligible overhead when compared with eventual consistency.

%% file: sections/intro.tex

\section{Introduction} \label{section:intro}

{\emph{Data Consistency Models.}}
To achieve high availability and low latency,
distributed data stores often geographically replicate data
at multiple sites called replicas~\cite{RDT:POPL2014, PoEC:NOW2014}.
However, this introduces the data consistency problem among replicas.
Over the past forty years, more than 50 consistency models
have been proposed~\cite{Baseball:CACM2013, Consistency:CSUR2016},
ranging from linearizability~\cite{Lin:TOPLAS1990}
to eventual consistency~\cite{EC:CACM2009, PoEC:NOW2014}.
According to the CAP theorem~\cite{CAP:SIGACTNews2002, CAP:PODC2000}
and the PACELC tradeoff~\cite{PACELC:Computer2012},
there are fundamental tradeoffs among consistency, availability,
and latency in the presence of network partition.
Therefore, no a one-size-fits-all consistency model exists~\cite{Baseball:CACM2013}.

{\emph{Tunable consistency.}}
To meet the needs of different applications,
many popular distributed data stores begin to provide
tunable consistency~\cite{DynamoDB:SOSP2007, Cassandra:SIGOPS2010, MongoDB:VLDB2019, CosmosDB:2018},
allowing clients to specify the consistency level per individual operation.
Amazon DynamoDB provides eventual consistency as default
and strong consistency with \texttt{ConsistentRead}~\cite{DynamoDB:SOSP2007}.
Apache Cassandra offers a number of fine-grained read and write consistency levels,
such as \texttt{ANY}, \texttt{ONE}, \texttt{Quorum}, and \texttt{ALL}~\cite{Cassandra:SIGOPS2010}.\footnote{
How is the consistency level configured? \url{https://docs.datastax.com/en/cassandra-oss/3.0/cassandra/dml/dmlConfigConsistency.html}}
For example, a read of level \texttt{Quorum} returns the value
after a quorum of replicas have responded.
MongoDB provides tunable consistency by exposing
the \emph{writeConcern} and \emph{readConcern} parameters
that can be set per operation~\cite{MongoDB:VLDB2019}.
For example, a read with \emph{readConcern} = \texttt{majority}
guarantees that the returned value has been written to a majority of replicas.
Azure Cosmos DB offers five well-defined consistency levels for read operations,
namely, from strongest to weakest, Strong, Bounded staleness, Session, Consistent prefix,
and Eventual~\cite{CosmosDB:2018}.\footnote{Consistency levels in Azure Cosmos DB. \url{https://docs.microsoft.com/en-us/azure/cosmos-db/consistency-levels}}
For example, Strong consistency offers linearizability~\cite{Lin:TOPLAS1990},
while Eventual consistency provides no ordering guarantees for reads.

{\emph{Causal Consistency and Session Guarantees.}}
Causal consistency guarantees that an update does not
become visible to clients until all its causal dependencies
are visible~\cite{CM:DC1995, PoEC:NOW2014, CC:PPoPP2016,
COPS:SOSP2011, GentleRain:SOCC2014, Eiger:NSDI2013, Cure:ICDCS2016, MongoDB:VLDB2019}.
Consider the classic ``Lost-Ring'' example~\cite{COPS:SOSP2011}.
Alice first posts ``I lost my ring'', and then posts ``I have found it'' after a while.
Bob sees both posts of Alice, and comments ``Glad to hear it''.
Causal consistency can avoid the undesired situation that
Charlie sees the comment by Bob but do not see the second post of Alice.
Causal consistency has been shown equivalent to
the conjunction of the four well-known \emph{session guarantees}~\cite{Session:PDP2004},
i.e., read your writes (\ryw), monotonic reads (\mr),
monotonic writes (\mw), and writes follow reads (\wfr)~\cite{SESSION:PDIS1994}.
\ryw{} ensures that any write becomes
visible to the subsequent reads in the same session.
\mr{} requires successive reads on the
same session observe monotonically increasing sets of writes.
\mw{} requires writes on the same
session take effect in the session order.
Finally, \wfr{} establishes causality between
two writes $w_1$ and $w_2$ via a read $r_1$,
if $r_1$ reads from $w_1$ and $w_2$ follows $r_1$ on the same session.

{\emph{Tunable Causal Consistency: Motivation}.}
Terry has demonstrated that it is desirable for different participants
of a baseball game to maintain the score with different session guarantees~\cite{Baseball:CACM2013}.
For example, \ryw{} is sufficient for the official scorekeeper
to retrieve the latest score before producing a new one,
while a radio reporter may need \mr{} to ensure that the observed baseball scores are monotonically increasing.
Therefore, it would be beneficial for a distributed data store to provide tunable causal consistency,
i.e., allowing clients to choose the session guarantee per individual operation.

{\emph{Tunable Causal Consistency: Related Work}.}
As far as we know, NuKV is the only distributed key-value store
that provides tunable causal consistency~\cite{Raft:DCN2019}.
However, it has two major drawbacks in terms of specification and implementation.
\begin{itemize}
  \item \emph{Drawback in Specification.}
    NuKV lacks a formal specification of tunable causal consistency.
	Instead, NuKV uses implementation itself as specification.
	Specifically, it first defines three write sets to track
	the writes committed at each server,
	the writes issued by each client,
	and the writes observed by each client, respectively.
	Then the four session guarantees are defined as different constraints
	on these three write sets.
  \item \emph{Drawback in Implementation.}
    NuKV implements \emph{per-key} session guarantees.
	It keeps track of session guarantees for each partition,
	being able to avoid the problem of slowdown cascades cross keys
	maintained by different partitions.
	However, per-key session guarantees are insufficient for some applications.
	For example, the ``Lost-Ring'' example mentioned above
	involves two keys representing Alice's and Bob's posts, respectively.
	To avoid the undesired situation,
	causality should be established across these two keys.
\end{itemize}

In this paper, we propose a formal specification and a general implementation of
tunable causal consistency (\tcc).
\begin{itemize}
  \item {\emph{Our First Contribution: Specification (Section~\ref{section:tcc}).}}
	We formally specify \tcc{} in an extended $(vis, ar)$ framework
	originally proposed by Burckhardt et al.~\cite{RDT:POPL2014, PoEC:NOW2014}.
	The visibility relation $vis$ specifies, for each operation,
	the set of operations that are visible to it.
	The arbitration relation $ar$ indicates
	how the system resolves the conflicts due to concurrent operations
	that are not visible to each other.
	To specify \tcc, we extend the $(vis, ar)$ framework by
	adding a \emph{consistency level} for each operation.
	Then we formalize \tcc{} in two steps.
	First, we define the visibility relation for each individual session guarantee.
	Then we consider the multi-level constraints,
	which specify how one session guarantee influences another
	in terms of the visibility and arbitration relations.
	We find that in \tcc{} any two session guarantees do \emph{not} impose any constraints on each other.
  \item \emph{Our Second Contribution: Implementation (Sections~\ref{section:algorithm} and \ref{section:exp}).}
	We design a \tcc{} protocol which is general in the sense that it supports multiple keys.
	The protocol uses vector clocks to track dependencies for each kind of session guarantees.
	Then, we implement a prototype distributed key-value store called \tccstore,
	which provides \tcc{} in the common sharded cluster deployment.
	We evaluate the performance of \tccstore{} on Aliyun.\footnote{Aliyun. \url{https://cn.aliyun.com/index.htm}}
	The latency is less than 38ms for all workloads and
	the throughput is up to about 2800 operations per second.
	We also show that TCC achieves better performance
	than causal consistency and requires a negligible overhead when
	compared with eventual consistency.
\end{itemize}

%% file: sections/tcc.tex

\section{Tunable Causal Consistency} \label{section:tcc}

In this section we formally specify \tcc{} in the $(vis, ar)$ framework~\cite{RDT:POPL2014, PoEC:NOW2014}.
We follow the description of the $(vis, ar)$ framework in~\cite{Framework:SRDS2020}
and extend it to support tunable consistency as needed.
As discussed in Section~\ref{section:intro},
we consider the set $\Level \triangleq \set{\ryw, \mr, \mw, \wfr}$ of four session guarantees.
We distinguish the set $\Level_{r} \triangleq \set{\ryw, \mr}$
of two session guarantees constraining reads
from that $\Level_{w} \triangleq \set{\mw, \wfr}$
of the other two session guarantees constraining writes.
\input{sections/tcc-vis-ar}
\input{sections/tcc-levels}

%% file: sections/tcc-vis-ar.tex

\subsection{Relations and Orderings} \label{ss:relations}

Given a set $A$, a binary relation $R$ over $A$ is a subset of $A \times A$,
i.e., $R \subseteq A \times A$.
For $a, b \in A$, we use $a \rel{R} b$ to denote $(a, b) \in R$.
We use $R^{-1}$ to denote the inverse relation of $R$,
i.e., $(a, b) \in R \iff (b, a) \in R^{-1}$.
We define $R^{-1}(b) \triangleq \set{a \in A \mid (a, b) \in R}$.

For two relations $R$ and $S$ over $A$, their composition is
$R \comp S \triangleq \set{(a, c) \mid \exists b \in A: a \rel{R} b \land b \rel{S} c}$.
For some subset $A' \subseteq A$, the restriction of $R$ to $A'$ is
$R|_{A'} \triangleq R \cap (A' \times A)$.
Let $f: A \to B$ be a function from $A$ to $B$ and $A' \subseteq A$.
The restriction of $f$ to $A'$ is
$f|_{A'} \triangleq f \cap (A' \times B) = \set{(a, f(a)) \mid a \in A'}$.

A relation is called a (strict) partial order when it is irreflexive and transitive.
A relation which is a partial order and total is called a total order.
\subsection{Read/Write Registers} \label{ss:registers}

We focus on the key-value store which maintains
a collection of integer read/write (or named get/put) registers.
An integer (read/write) register supports two operations:
$\wr(v)$ writes value $v \in \mathbb{Z}$ to the register,
and $\rd$ reads value from the register.
We use $\bot$ to indicate that writes return no values.
Let $\Op = \set{\wr, \rd}$ and $\Val = \mathbb{Z} \cup \set{\bot}$.

The sequential semantics of registers is defined by
a function $\eval_{\intreg}: \Op^{\ast} \times \Op \to \Val$ that,
given a sequence of operations $S$ and an operation $o$,
determines the return value $\eval_{\intreg}(S, o) \in \Val$ for $o$
when $o$ is performed after $S$~\cite{PoEC:NOW2014}.
If $o$ is a $\rd$ operation, it returns the value of the last preceding $\wr$,
or the initial value 0 if there are no prior writes~\cite{Framework:SRDS2020}.
Formally, for any operation sequence $S$,
\begin{align*}
\eval_{\intreg}(S, \wr(v)) &= \bot, \\
\eval_{\intreg}(S, \rd) &= v,
	\text{ if } \wr(0)\;S = S_{1}\; \wr(v)\; S_{2} \\
	&\qquad\;\; \text{and } S_{2} \text{ contains no } \wr \text{ operations}.
\end{align*}
\subsection{Histories}  \label{ss:history}

Clients interact with the key-value store by performing operations on keys.
The interactions visible to clients are recorded in a \emph{history}.
To support tunable consistency,
we tag each operation with a consistency level.

\begin{definition}[Histories]  \label{def:history}
  A {\it history} is a tuple $H = (E, \op, \level, \rval, \so)$ such that
  \begin{itemize}
    \item $E$ is the set of all events of operations
      invoked by clients in a single computation;
    \item $\op: E \to \Op$ describes the operation of an event;
    \item $\level: E \to \Level$ specifies the \emph{consistency level}
	  requested by the operation $\op(e)$ of an event $e$;
    \item $\rval: E \to \Val$ describes the value returned by
      the operation $\op(e)$ of an event $e$;
    \item $\so \subseteq E \times E$ is a partial order over $E$, called the \emph{session order}.
      It relates operations within a session in the order they were invoked by clients.
  \end{itemize}
\end{definition}

For a history $H = (E, \op, \level, \rval, \so)$, we define:
\begin{itemize}
  \item $E_r$ is the set of events of all read operations in $H$.
  \item $E_w$ is the set of events of all write operations in $H$.
  \item For a level $l \in L$,
    $E_{l} \triangleq E_w \;\cup\; \{e\in E_r \mid \level(e) = l\}$
    is the set of events of all write operations
    and the read operations with consistency level $l \in \Level$.
  \item For a level $l \in L$, we use $H_{l}$
    to denote the restriction of $H$ to the events $E_{l}$,
	i.e., $H_{l} \triangleq (E_{l}, \op|_{E_{l}}, \level|_{E_{l}}, \rval|_{E_{l}}, \so|_{E_{l}})$.
\end{itemize}

\begin{figure}[t]
  \centering
  \includegraphics[width = 0.45\textwidth]{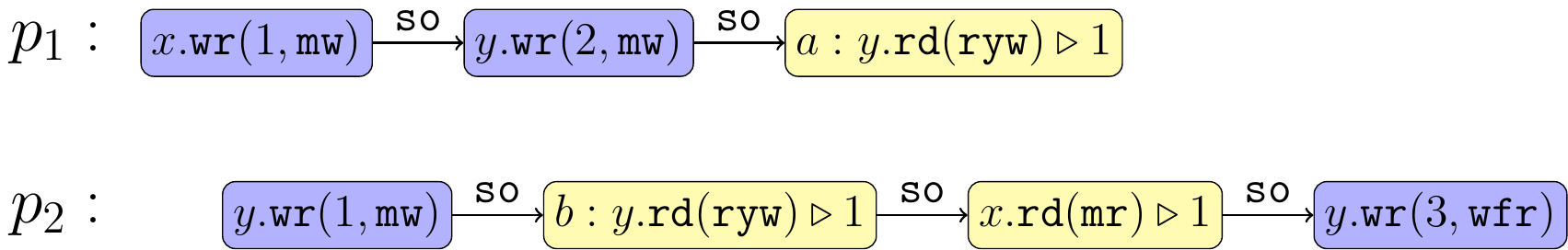}
  \caption{A history consisting of two sessions $p_1$ and $p_2$ and two registers $x$ and $y$.
    Here $r.\wr(v,l) \triangleright \bot$ denotes the operation of writing value $v$
    to register $r$ with consistency level $l$
    (the return value $\bot$ is omitted).
    Similarly, $r.\rd(l) \triangleright v$ denotes the operation of reading value $v$
    from register $r$ with consistency level $l$.
    We use labels, such as $a$ and $b$, to make events unique.}
  \label{fig:exp-register-history}
\end{figure}

\begin{example} 
  \label{ex:history}
  Consider the history in Fig.~\ref{fig:exp-register-history}
  consisting of two sessions $p_1$ and $p_2$ and two registers $x$ and $y$.
  We have $E_{\ryw} = \set{x.\wr(1,\mw), y.\wr(1,\mw), y.\wr(3,\wfr),
  a:y.\rd(\ryw) \triangleright 1, b:y.\rd(\ryw) \triangleright 1}$.
  Moreover, $\so|_{E_{\ryw}}$ is the session order over $E_{\ryw}$,
  including $x.\wr(1,\mw) \rel{\so} y.\wr(2,\mw) \rel{\so} a:y.\rd(\ryw)\triangleright1$
  and $y.\wr(1,\mw) \rel{\so} b:y.\rd(\ryw)\triangleright1 \rel{\so} y.\wr(3,\wfr)$.
\end{example}

\begin{figure}[t]
  \centering
  \includegraphics[width = 0.48\textwidth]{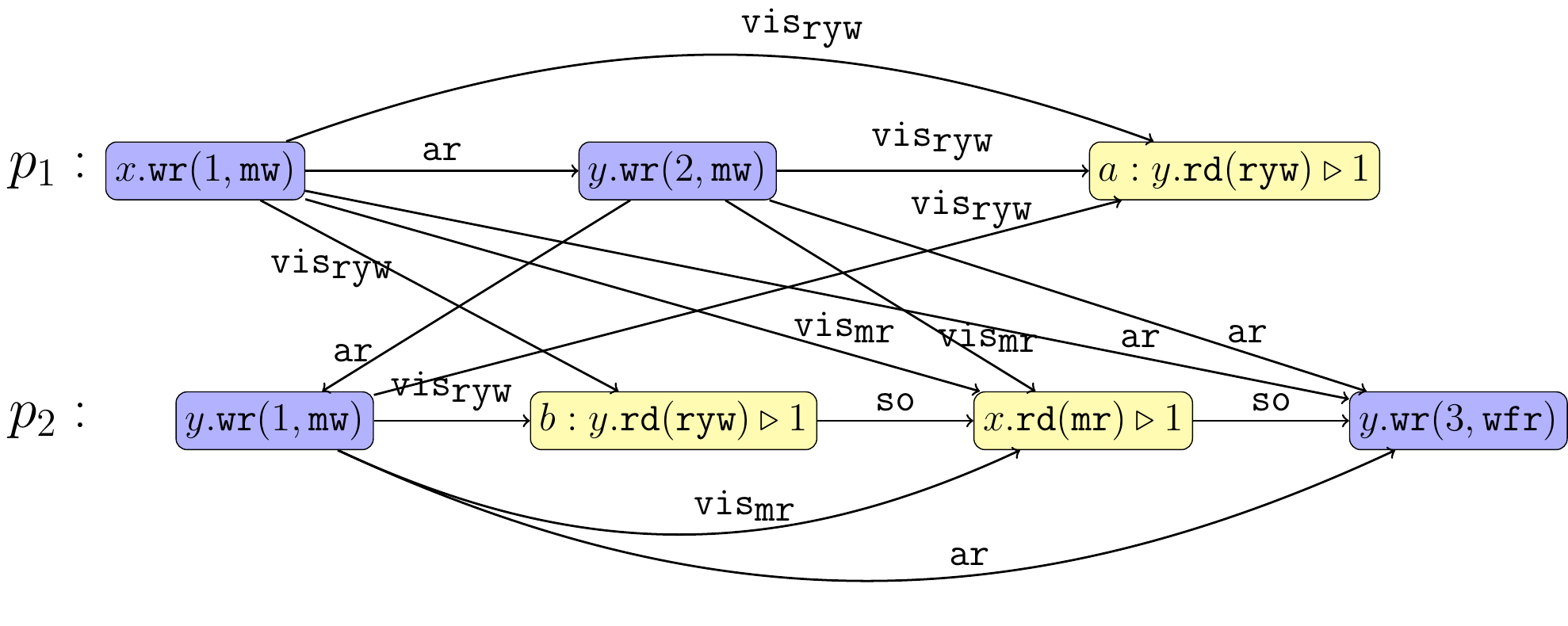}
  \caption{An abstract execution of the history in Fig.~\ref{fig:exp-register-history}.
  Assume that $x.\wr(1,\mw)$ has been applied on replica $R_1$
  when $b:y.\rd(\ryw)$ is executed on $R_1$,
  $y.\wr(2,\mw)$ has been applied on $R_1$
  when $x.\rd(\mr)$ is executed on $R_1$ and
  $y.\wr(2,\mw)$ is applied before $y.\wr(1,\mw)$ on $R_1$.}
  \label{fig:exp-register-execution}
\end{figure}
\subsection{Abstract Executions}  \label{ss:ae}

To justify the return value of an event in a history,
we need to know the set of events that are visible to it
and how these events are ordered.
These are captured declaratively by the \emph{visibility}
and \emph{arbitration} relations, respectively~\cite{RDT:POPL2014, PoEC:NOW2014}.

\begin{definition}[Abstract Executions] \label{def:ae}
  An \emph{abstract execution} is a triple
  $A = ((E, \op, \level, \rval, \so), \vis, \ar)$ such that
  \begin{itemize}
    \item $(E, \op, \level, \rval, \so)$ is a history;
    \item Visibility $\vis \triangleq \bigcup_{l \in \Level_{r}} \vis_{l}$ is an acyclic relation,
	  where $\vis_{l} \subseteq E_{l} \times E_{l}$
	  is the visibility relation for consistency level $l \in L_{r}$;
    \item Arbitration $\ar \subseteq E \times E$ is a total order such that $\vis \subseteq \ar$.
  \end{itemize}
\end{definition}

\begin{example} 
  \label{ex:ae}
  Fig.~\ref{fig:exp-register-execution} shows an abstract execution
  of the history in Fig.~\ref{fig:exp-register-history}.
  Assume that $x.\wr(1,\mw) \rel{\vis_{\ryw}} b:y.\rd(\ryw) \triangleright 1$,
  $y.\wr(2,\mw) \rel{\vis_{\mr}} x.\rd(\mr)\triangleright 1$,
  and $y.\wr(2,\mw) \rel{\ar} y.\wr(1,\mw)$.
  In the following, we informally explain the other visibility
  and arbitration relations required by \tcc.
  \begin{itemize}
    \item For event $b:y.\rd(\ryw)\triangleright1$, \ryw{} requires that
      all writes placed before it in $p_2$ are visible to it.
      To satisfy \ryw, we add the visibility relation
      $y.\wr(1,\mw) \rel{\vis_{\texttt{ryw}}} b:y.\rd(\ryw)\triangleright1$.
      Similarly, we add the other $\vis_{\ryw}$ relations as shown in Fig.~\ref{fig:exp-register-execution}.
    \item For event $x.\rd(\mr)\triangleright1$, \mr{} requires that
      it sees all writes visible to $b:y.\rd(\ryw)\triangleright1$.
      To satisfy \mr, we add the visibility relation
      $y.\wr(1,\mw) \rel{\vis_{\mr}} x.\rd(\mr)\triangleright1$ and
      $x.\wr(1,\mw) \rel{\vis_{\mr}} x.\rd(\mr)\triangleright1$.
    \item For event $y.\wr(2,\mw)$, \mw{} requires that
      it should be performed after $x.\wr(1,\mw)$.
      To satisfy \mw, we add the arbitration relation
      $x.\wr(1,\mw) \rel{\ar} y.\wr(2,\mw)$.
    \item For event $y.\wr(3,\wfr)$, \wfr{} requires that
      it should be performed after all the events visible to $x.\rd(\mr)\triangleright1$.
      To satisfy \wfr, we add the arbitration relations
      $x.\wr(1,\mw) \rel{\ar} y.\wr(3,\wfr)$, $y.\wr(2,\mw) \rel{\ar} y.\wr(3,\wfr)$,
      and $y.\wr(1,\mw) \rel{\ar} y.\wr(3,\wfr)$.
  \end{itemize}
\end{example}
\subsection{Consistency Models} \label{ss:consistency-models}

\begin{definition}[Consistency Models] \label{def:cm}
  A \emph{consistency model} is a set of consistency predicates on abstract executions.
\end{definition}

We write $A \models P$ if the consistency predicate $P$ is true
on the abstract execution $A$.

\begin{definition}[Satisfication (Abstract Execution)] \label{def:sat-ae}
  An abstract execution $A$ \emph{satisfies} consistency model
  $\mathcal{M} = \set{P_1, \dots, P_n}$, denoted $A \models \mathcal{M}$,
  if $A$ satisfies each consistency predicate in $\mathcal{M}$.
  That is, $A \models \mathcal{M} \iff A \models P_1 \land \dots \land A \models P_n$.
\end{definition}

We focus on histories that satisfy some consistency model.

\begin{definition}[Satisfication (History)] \label{def:sat-history}
  A history $H$ \emph{satisfies} consistency model
  $\mathcal{M} = \set{P_1, \dots, P_n}$, denoted $H \models \mathcal{M}$,
  if it can be extended to an abstract execution 
  that satisfies $\mathcal{M}$.
  That is, $H \models \mathcal{M} \iff \exists\; \vis, \ar.\; (H, \vis, \ar) \models \mathcal{C}$.
\end{definition}
\subsection{Return Value Consistency}  \label{ss:retval}

A common consistency predicate is the consistency of return values.
In an abstract execution $A$, the return value of an event $e$
is determined by its \emph{operation context}, denoted $\ctxt_{A}(e)$,
which is the restriction of $A$ to the set $\vis_{\level(e)}^{-1}(e)$
of events visible to $e$ with respect to consistency level $\level(e)$.
Since \ar{} is a total order,
the events in $\ctxt_{A}(e)$ can be ordered into
a sequence to justify $\rval(e)$.

\begin{definition}[Operation Context] \label{def:ctxt-revisited}
  Let $A = ((E, \op, \level, \rval, \so), \vis, \ar)$ be an abstract execution.
  The \emph{operation context} of $e \in E$ in $A$ is defined as
  \[
    \ctxt_{A}(e) \triangleq A|_{{\vis_{\level(e)}^{-1}(e),\op, \vis, \ar}}.
  \]
\end{definition}

Accordingly, the return value consistency (\retval{}) predicate is defined as follows.

\begin{definition}[Return Value Consistency] \label{def:retval-revisited}
  For read/write registers \intreg, the return value consistency predicate
  on an abstract execution $A$ is
  \[
    \retval(\intreg) \triangleq \forall e \in E.\;
      \rval(e) = \eval_{\intreg} \big( \ctxt_{A}, \op(e) \big).
  \]
\end{definition}

\begin{example} 
  \label{exp:rval}
  Consider the abstract execution in Fig.~\ref{fig:exp-register-execution}.
  We show how to justify the return value of the event $a: y.\rd(\ryw)\triangleright 1$.
  First, $\vis_{\ryw}^{-1}(a) = \set{x.\wr(1,\mw), y.\wr(1,\mw), y.\wr(1,\mw)}$.
  Second, we have $x.\wr(1,\mw) \rel{\ar} y.\wr(2,\mw) \rel{\ar} y.\wr(1,\mw)\rangle$.
  Therefore, $a$ can be justified by the operation sequence
  $\langle x.\wr(1,\mw) \; y.\wr(2,\mw) \; y.\wr(1,\mw)\rangle$.
  Similarly, the read events $b: y.\rd(\ryw)\triangleright1$ and
  $x.\rd(\mr)\triangleright 1$ can be justified by
  the operation sequence $\langle x.\wr(1,\mw) \; y.\wr(1,\mw)\rangle$ and
  $\langle x.\wr(1,\mw) \;  y.\wr(2,\mw) \; y.\wr(1,\mw)\rangle$, respectively.
\end{example}

%% file: sections/tcc-levels.tex

\subsection{Individual Level Constraints}  \label{ss:individual}

We define \tcc{}
in the above extended $(vis, ar)$ framework in two steps:
In this section we define the constraints
specified by the four session guarantees individually.
In the next section we define the multi-level constraints,
specifying how one session guarantee influence another
in terms of the visibility and arbitration relations.

\subsubsection{Read Your Writes} \label{sss:ryw}

\ryw{} ensures that any write becomes visible to the subsequent reads in the same session.

\begin{definition}[Read Your Writes (\ryw)] \label{def:ryw}
  Let $H$ be a history. The constraint $\mathcal{C}^{\ryw}$ for \ryw{} is
  \begin{align*}
    \mathcal{C}^{\ryw} \triangleq \; &\forall e \in E_r, e' \in E_w.\\
      & (\level(e) = \ryw \land e' \rel{\so} e) \implies e' \rel{\vis_{\ryw}} e.
  \end{align*}

Note that write events on other sessions ($\centernot{\rel{\so}} e$)
can also be visible ($\vis_{\ryw}$) to $e$.
For example, as shown in Fig.~\ref{fig:exp-register-execution},
$x.\wr(1,\mw)$ is visible ($\vis_{\ryw}$) to $b:y.\rd(\ryw)\triangleright1$.

\end{definition}
\begin{example} 
  \label{ex:ryw}
  Consider the history in Fig.~\ref{fig:exp-register-execution}.
  For event $b:y.\rd(\ryw) \triangleright 1$ with consistency level \ryw,
  since $y.\wr(1,\mw) \rel{\so} b:y.\rd(\ryw) \triangleright 1$,
  to satisfy $\mathcal{C}^{\ryw}$,
  we have $y.\wr(1,\mw) \rel{\vis_{\ryw}} b:y.\rd(\ryw) \triangleright 1$.
  Similarly, we have $y.\wr(2,\mw) \rel{\vis_{\ryw}} a:y.\rd(\ryw) \triangleright 1$
  and $x.\wr(1,\mw) \rel{\vis_{\ryw}} a:y.\rd(\ryw) \triangleright 1$
  due to $y.\wr(2,\mw) \rel{\so} a:y.\rd(\ryw) \triangleright 1$
  and $x.\wr(1,\mw) \rel{\so} a:y.\rd(\ryw) \triangleright 1$, respectively.
\end{example}
\subsubsection{Monotonic Reads} \label{sss:mr}

\mr{} requires successive reads on the same session
observe monotonically increasing sets of writes.

\begin{definition}[Monotonic Reads (\mr)] \label{def:mr}
Let $H$ be a history. The constraint $\mathcal{C}^{\mr}$ for \mr{} is
  \begin{align*}
    \mathcal{C}^{\mr} \triangleq \;
	  &\forall e_1 \in E_w, e_2 \in E_r, e \in E_r.\\
      &(\level(e) = \mr \land e_1 \rel{\vis} e_2 \rel{\so} e) \implies e_1 \rel{\vis_{\mr}} e.
  \end{align*}
Note that $\vis$ is either $\vis_{\ryw}$ or $\vis_{\mr}$,
and $e_1$ and $e_2$ can be on different sessions.
\end{definition}
\begin{example} 
  \label{ex:mr}
  Consider the history in Fig.~\ref{fig:exp-register-execution}.
  Both $x.\wr(1,\mw)$ and $y.\wr(1,\mw)$ are visible to $b:y.\rd(\ryw)\triangleright 1$.
  Since $b:y.\rd(\ryw)\triangleright 1 \rel{\so} x.\rd(\mr)\triangleright 1$,
  to satisfy $\mathcal{C}^{\mr}$,
  both $x.\wr(1,\mw)$ and $y.\wr(1,\mw)$
  should be visible (via $\vis_{\mr}$) to $x.\rd(\mr)\triangleright 1$.
\end{example}
\subsubsection{Monotonic Writes} \label{sss:mw}

\mw{} requires writes on the same session take effect in the session order.

\begin{definition}[Monotonic Writes (\mw)] \label{def:mw}
  Let $H$ be a history. The constraint $\mathcal{C}^{\mw}$ for \mw{} is
  \begin{align*}
    \mathcal{C}^{\mw} \triangleq \;&\forall e', e \in E_w.\; \\
	  &(\level(e) = \mw \land e' \rel{\so} e) \implies e' \rel{\ar} e.
  \end{align*}
\end{definition}
\begin{example} 
  \label{ex:mw}
  Consider the history in Fig.~\ref{fig:exp-register-execution}.
  Since $x.\wr(1,\mw) \rel{\so} y.\wr(2,\mw)$,
  $\mathcal{C}^{\mw}$ requires $x.\wr(1,\mw) \rel{\ar} y.\wr(2,\mw)$.
\end{example}
\subsubsection{Writes Follow Reads} \label{sss:wfr}

\wfr{} establishes causality between two writes $w_{1}$ and $w_{2}$ via a read $r_{1}$,
if $r_{1}$ reads from $w_{1}$ and $w_2$ follows $r_{1}$ on the same session.

\begin{definition}[Writes Follow Reads (\wfr)] \label{def:wfr}
  Let $H$ be a history.
  The constraint $\mathcal{C}^{\wfr}$ for \wfr{} is
  \begin{align*}
    \mathcal{C}^{\wfr} \triangleq \;&\forall e_1\in E_w, e_2 \in E_r, e \in E_w.\\
      &(\level(e) = \wfr \land e_1 \rel{\vis} e_2 \rel{\so} e) \implies e_1 \rel{\ar} e.
  \end{align*}
Similar to Definition~\ref{def:mr},
$\vis$ is either $\vis_{\ryw}$ or $\vis_{\mr}$,
and $e_1$ and $e_2$ can be on different sessions.
\end{definition}
\begin{example} 
  \label{ex:wfr}
  Consider the history in Fig.~\ref{fig:exp-register-execution}.
  Since $x.\wr(1,\mw) \rel{\vis_{\texttt{ryw}}} b:y.\rd(\ryw) \triangleright 1$
  and $b:y.\rd(\ryw)\triangleright 1 \rel{\so} y.\wr(3,\wfr)$,
  \wfr{} requires $x.\wr(1,\mw) \rel{\ar} y.\wr(3,\wfr)$.
\end{example}
\subsection{Multi-level Constraints} \label{ss:multi-level}

Intuitively, it is not sufficient to
separately satisfy the constraints corresponding to individual levels.
It is also necessary to specify how one level influences another.
For example, Bouajjani et al. provide multi-level constraints
for strong and weak consistency levels in~\cite{Multilevel:VMCAL2020}.

In this section we define the multi-level constraints,
which specify how one session guarantee influence another
in terms of the visibility and arbitration relations.
Consider two events $e'$ of level $l'$ and $e$ of level $l$
such that $e$ immediately follows $e'$ on the same session.
We use $\psi_{l'}^{l}$ to denote the constraint that $l'$ imposes on $l$.

To find the constraints imposed by one session guarantee on another,
we enumerate all pairs of session guarantees.
However, we find that any multi-level constraint has been enforced by
a certain constraint for individual session guarantee.

\begin{definition}[Multi-level Constraints]\label{def:multi-constraint}
  The multi-level constraint for any pair of session guarantees $l', l \in L$
  is $\psi_{l'}^{l} \triangleq \top$.
\end{definition}
That is, $l'$ does not impose any constraints on $l$.
We explain it by example.

\begin{example} 
  \label{exp:multilevel}
  For example, if $l = \ryw$, $e$ needs to observe
  all the writes preceding it on the same session,
  regardless of the level of $e'$.
  Suppose $e'$ is a read and $l = \wfr$.
  The event $e$ should be ordered after all the writes visible to $e'$ in $\ar$.
  This requirement is enforced by $\mathcal{C}^{\wfr}$.
\end{example}

\begin{definition}[Tunable Causal Consistency (\tcc)] \label{def:tcc}
  \[
    \textsc{\tcc} \triangleq \retval(\texttt{reg})
	  \land \bigwedge_{l \in L} \mathcal{C}^{l}
	  \land \bigwedge_{l', l \in L} \psi_{l'}^{l}.
  \]
\end{definition}

\begin{example} 
  \label{ex:tcc}
  According to Examples~\ref{exp:rval}$\sim$\ref{exp:multilevel},
  the history in Fig.~\ref{fig:exp-register-execution} satisfies \tcc.
\end{example}

%% file: sections/curelc.tex

\section{Protocol} \label{section:algorithm}

We assume a distributed key-value store that maintains multiple data items.
The data is fully replicated across different datacenters.
Inside each datacenter, data is partitioned into several partitions.
All datacenters adopt the same partitioning strategy.
As shown in Fig.~\ref{fig:archi}, we consider a configuration
consisting of $D$ datacenters, each of which consists of $N$ partitions.
We assume that each partition is equipped with a physical clock,
which generates increasing timestamps.
Clocks are loosely synchronized by the classic time synchronization protocol NTP.\footnote{NTP. \url{http://www.ntp.org}}
Client operations can be executed on their local datacenters or remote ones.

\begin{figure}[t]
  \centering
  \includegraphics[width = 0.35\textwidth]{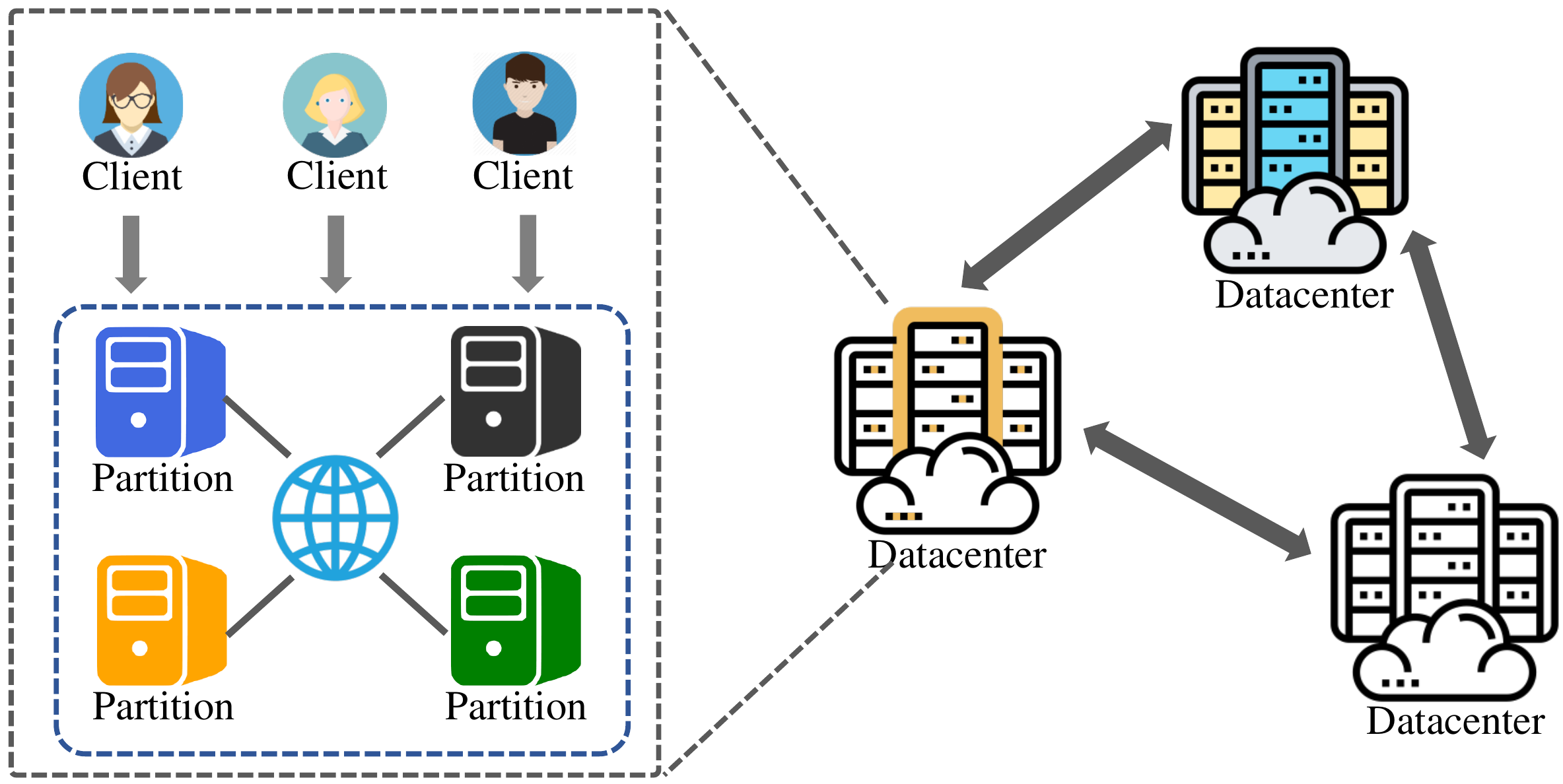}
  \caption{System architecture}\label{fig:archi}
\end{figure}

We assume the data store keeps multiple versions of each data item.
An update on a data item creates a new version of it.
Our protocol offers two operations to the clients:

\begin{enumerate}
  \item $\pput(k,v,l)$: A \pput{} operation assigns value $v$ to an item
        identified by key $k$ while ensuring consistency level $l$.
        If the item does not exist, a new item with value $v$ is created.
        If key $k$ exists, a new version with value $v$ is created.
  \item $\get(k,l)$: A \get{} operation returns the value of the item
        identified by key $k$ while ensuring consistency level $l$.
\end{enumerate}
\subsection{States}\label{ss:states}

Table~\ref{table:notations} provides a summary of notations used in the protocol.

\input{tables/notations}
\subsubsection{Client States} \label{sss:client-states}

Each client $c$ maintains two vector clocks:
the highest read vector clock \hrvc{}
records the maximum timestamp of stable versions read by client $c$,
and the highest write vector clock \hwvc{}
is the maximum timestamp of versions written by client $c$.
The client also maintains two vector clocks $\cvc_{c}^{r}$ and $\cvc_{c}^{w}$
that record the maximum timestamps of versions
read and written by the client, respectively.
\subsubsection{Server States} \label{sss:client-states}

For each partition $p_{d}^{m}$, $\clock_{d}^{m}$ records the value its physical clock.
The server maintains a vector clock $\pvc_{d}^{m}$ of size $D$,
where $\pvc_{d}^{m}[k]$ indicates that $p_{d}^{m}$ has received updates
up to $\pvc_{d}^{m}[k]$ from partition $p_{k}^{m}$.
The server also maintains a stable vector $\css_{c}^{m}$ of size $D$ to
denote the latest globally stable consistent view
of the local datacenter known by $p_{d}^{m}$,
i.e., the view that $p_{d}^{m}$ knows to be available
at all partitions in the local datacenter.
To advance $\css_{d}^{m}$,
partitions in the same datacenter periodically exchange their $\pvc$.
\subsection{Protocol} \label{sss:protocol}

We now informally describe how the \pput{} and \get{} operations
are executed at clients and servers.
We also describe the clock management and replication mechanism.
The pseudocode and correctness proof can be found in Appendix of~\ref{appendix:protocol} and \ref{appendix:proof}.
\subsubsection{$\get(k,l)$} \label{sss:get}

A client $c$ sends a \get{} request, containing the key $k$,
and two vectors $vc_{r}$ and $vc_{w}$
to a server which stores key $k$.
Let $O$ be a $D$ dimensional vector with all entries equal to zero.
When the client sends a \get{} request,
it can use $O$ or $\hrvc$ as $vc_{r}$ to
  guarantee \ec{} (eventual consistency) or \mr{} respectively.
  Similarly, it can use $O$ or $\hwvc$ as $vc_{w}$ to
  guarantee \ec{} or \ryw{} respectively.

When the server $p_d^m$ receives the \get{} request,
it first checks whether the dependencies specified by $vc_r$ and $vc_w$
have been applied locally,
by comparing its $\css$ with $vc_r$ and $vc_w$.
The server blocks if for some $i$ ($i\in D$),
$vc_r[i]$ or $vc_w[i]$ is greater than $\css_{d}^{m}[i]$.
When both $vc_r$ and $vc_w$ are no greater than $\css_{d}^{m}$,
the server retrieves the latest stable version
in the version chain of requested key $k$,
which has an update timestamp no greater than the server's $\css$.
Finally, the value $v$ of the version, its update timestamp,
and $\css_{d}^{m}$ are returned to the client.
Upon receiving the reply,
the client updates its $\hrvc$ and $\cvc_{c}^{r}$ accordingly.

\subsubsection{$\pput(k,v,l)$} \label{sss:put}

A client $c$ sends a \pput{} request, containing the key $k$,
the value $v$, and its dependency time denoted by $dvc$
to the server which stores key $k$.
The client chooses different $dvc$ to provide different session guarantees.
Choosing $O$, $\cvc_{c}^w$ or $\cvc_{c}^r$ guarantees
eventual consistency, \mw{} or \wfr{}, respectively.

When the server $p_d^m$ receives the \pput{} request,
it first checks whether the $d$-th entry of dependency time $dvc$
sent by client is smaller than local physical clock $\clock_{d}^{m}$.
If not, the server will wait until the condition becomes true.
Then, the server updates the $d$-th entry of $dvc$
and $\pvc_{d}^{m}$ with $\clock_{d}^{m}$.
The server creates a new version of the item identified by key $k$
by assigning it a tuple consisting of the key $k$, the value $v$,
and the update timestamp $dvc$, and
inserts it into the version chain of the item
and the set of updates which are waited to be replicated.
Finally, the server sends a reply with the newly created $dvc$ to the client.
Upon receiving the reply,
the client updates its $\hwvc$ and $\cvc_{c}^{w}$ accordingly.
\subsubsection{$Replication$} \label{sss:replication}

Inside a datacenter, each partition periodically replicates the updates
in the set of $\updates$ in $vc[d]$ order
to its replicas at the other datacenters.
When there is no update to replicate, a heartbeat is sent,
which contains its latest clock time.
If a server receives a heartbeat from datacenter $i$,
it updates the $i$-th entry of its $\pvc$.
If it receives a replication request,
it additionally inserts the received new version into
the local corresponding version chain.
However, this update is not visible to clients until
the server's \css{} becomes larger than its update timestamp.
The servers inside each datacenter
exchange their \pvc{} vectors in the background,
and each server $p_d^m$ computes its $\css_{d}^{m}$
as the aggregate minimum of all known $\pvc$.

%% file: tables/notations.tex

\begin{table}[t]
  \centering
  \caption{Notations used in the protocol description}
  \label{table:notations}
  \renewcommand{\arraystretch}{1.1}
  \begin{tabular}{l||l}
    \hline
    Notations & Description \\
    \hline \hline
    $\hrvc$ & highest read vector clock \\
    $\hwvc$ & highest write vector clock \\
    $\cvc_{c}^{r}$ & vector clock for reads at client $c$ \\
    $\cvc_{c}^{w}$ & vector clock for writes at client $c$ \\
    \hline \hline
     $p^{m}_{d}$ & partition $m$ in data center $d$\\
    $\css_{d}^{m}$ & global stable vector clock at $p_{d}^{m}$ \\
    $\pvc_{d}^{m}$ & vector clock at $p_{d}^{m}$ \\
    $\clock_{d}^{m}$ & physical clock at $p_{d}^{m}$\\
    $\pmc^{m}_{d}$ & matrix of received $\pvc_{d}^{i}$ at $p_{d}^{m}$ \\
    \hline \hline
    $\partition(k)$ & the partition that holds key $k$\\
    \hline \hline
    $k$ & key \\
    $v$ & value \\
    $l$ & consistency level \\
    $vc$ & vector clock \\
    $dvc$ & dependency vector clock \\
    \hline
  \end{tabular}
\end{table}

%% file: sections/evaluation.tex

\section{Experimental Results} \label{section:exp}

We develop a prototype distributed key-value store called \tccstore{} providing \tcc.
In this section we evaluate its performance in terms of throughput and latency,
and investigate the effect of the locality of traffic, workload characteristics, and deployment.
For comparison, we also implement causal consistency and eventual consistency in \tccstore{}.
\subsection{Experimental Setup} \label{ss:exp-setup}

\tccstore{} is implemented in C++.
We use Google Protocol Buffers\footnote{
  Google Protocol Buffers. \url{https://developers.google.com/protocol-buffers/}} for communication.
We conduct our experiments in two datacenters on Aliyun,
located at ShenZhen in South China and HangZhou in East China, respectively.
The average RTT between them is about 27ms.
We run all replicas on machines running Ubuntu 16.04
with 2 vCPUs, 2.5 GHz Intel Xeon (Cascade Lake), 4 GB memory, and 40 GB storage.

We run different number of client threads to generate different load conditions.
We use one machine per datacenter
with the same specification as all replicas
to run client threads.
Our default workload uses the 50:50 read:write ratio and
runs operations on a platform deployed over 2 datacenters (3 partitions per datacenter).
Operations access keys within a partition
according to a uniform distribution.
\input{sections/exp-results}

%% file: sections/exp-results.tex

\input{figs/ratio1}

\input{figs/load1}

\subsection{Effect of Locality of Traffic} \label{ss:effect-locality}

In the first set of experiments, we investigate the effect of the locality of traffic
on the performance of \tccstore{}.
We consider six cases, $\ec/\ec$, $\cc/\cc$, $\ryw/\mw$,
$\ryw/\wfr$, $\mr/\mw$ and $\mr/\wfr$.
Each case represents the consistency levels that can be chosen by clients.
For example, in the $\ec/\ec$ (resp. $\cc/\cc$) case,
both levels for read and write operations
can only be eventual consistency (resp. causal consistency).
First, we consider these cases where
clients only access their local datacenters.

Fig.~\ref{fig:ratio-latency} shows the effect of write proportion
on the operation latency.
The results are for 36 client threads per datacenter.
As we expect, latency increases as the write proportion increases,
since write operations are more expensive than read operations.
Fig.~\ref{fig:ratio-throughput} shows how throughput changes as we change the write proportion.
In all diagrams with throughout, we report the number of
total operations done by the clients inside one datacenter.
As shown in Fig.~\ref{fig:ratio-throughput},
the throughput drops as we increase the write proportion.

We observe eventual consistency gains better performance compared with other cases.
However, the increased latency of providing any session guarantee
compared with eventual consistency is negligible (less than 1.5ms).
Causal consistency has the highest latency
and the lowest throughput in all cases.
It requires about up to 8\% higher latency
and has up to 6\% lower throughput
than other session guarantees.

Next, we take into account the cases with 25\% remote traffic.
Fig.~\ref{fig:remote-ratio-latency} and Fig.~\ref{fig:remote-ratio-throughput}
show the corresponding results.
In all cases, eventual consistency still has the best performance.
Unlike the cases where clients only access their local datacenters,
the difference between providing causal consistency
and session guarantees is obvious here.
Causal consistency requires about up to 10\% higher latency
and has up to 13\% lower throughput than other session guarantees.

\subsection{Effect of Workload} \label{ss:effect-load}

In this section, we want to see how the performance changes for
workloads with various number of client threads.
Fig.~\ref{fig:thread-latency} and Fig.~\ref{fig:thread-throughput}
show the effect of the number of client threads
on the operation latency and throughput, respectively.
We run 4, 8, 12, 16, and 20 client threads per partition
to generate different load conditions.
As we expect, both latency and throughput increase in all cases as
we run more client threads.

Again, we observe that eventual consistency has better performance,
and causal consistency still has the highest latency and lowest throughput in all cases.
It requires up to 34\% higher latency
and has up to 11\% lower throughput
than other session guarantees.

Next, we take into account the cases with 25\% remote traffic.
Fig.~\ref{fig:remote-thread-latency} and Fig.~\ref{fig:remote-thread-throughput}
show the corresponding results.
Again, we observe that the difference
between causal consistency and session guarantees is more obvious.
Causal consistency requires up to 40\% higher latency
and has up to 23\% lower throughput
than other session guarantees.

\input{figs/partition}
\subsection{Effect of Deployment} \label{ss:effect-deployment}

We examine how the performance of \tccstore{} changes with different number of partitions.
Fig.~\ref{fig:partition-latency} and Fig.~\ref{fig:partition-throughput} show the latency
and throughput achieved by \tccstore{} with 2, 3, and 5 partitions.
All the cases run 12 client threads per partition and
provide \mr{} level for reads and \wfr{} level for writes with 100\% local traffic.
We observe that \tccstore{} with 2 partitions gains better performance on latency.
However, more partitions gain better performance on throughput.
The throughput achieved with 5 partitions is about 2 and 1.5 times
higher than the throughput achieved with 2 and 3 partitions, respectively.

%% file: figs/ratio1.tex

\begin{figure*}[h]
  \centering
    \begin{minipage}[h]{0.49\linewidth}
    \centering
    \subfigure[Latency]{
      \includegraphics[width = 0.46\textwidth]{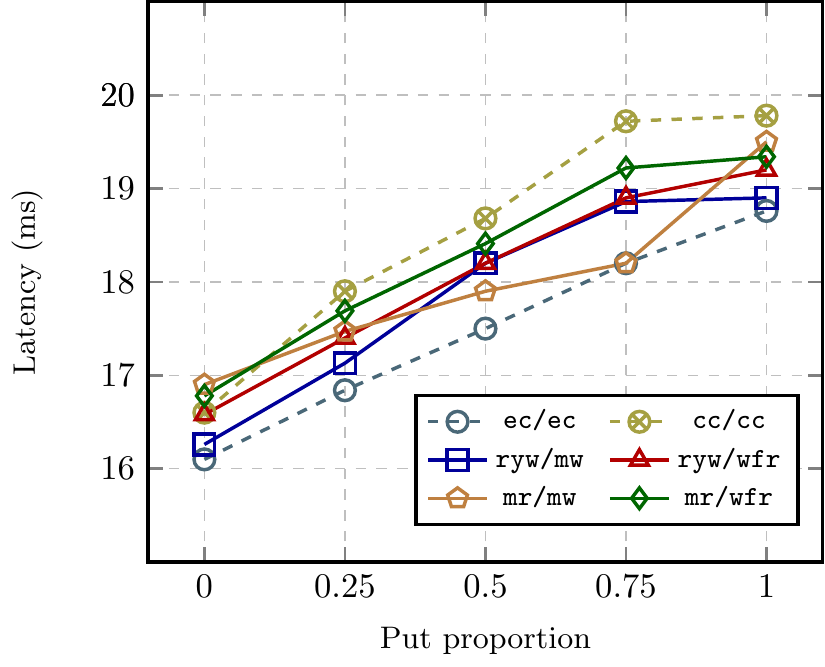}\label{fig:ratio-latency}
    }
    \subfigure[Throughput]{
      \centering
      \includegraphics[width = 0.46\textwidth]{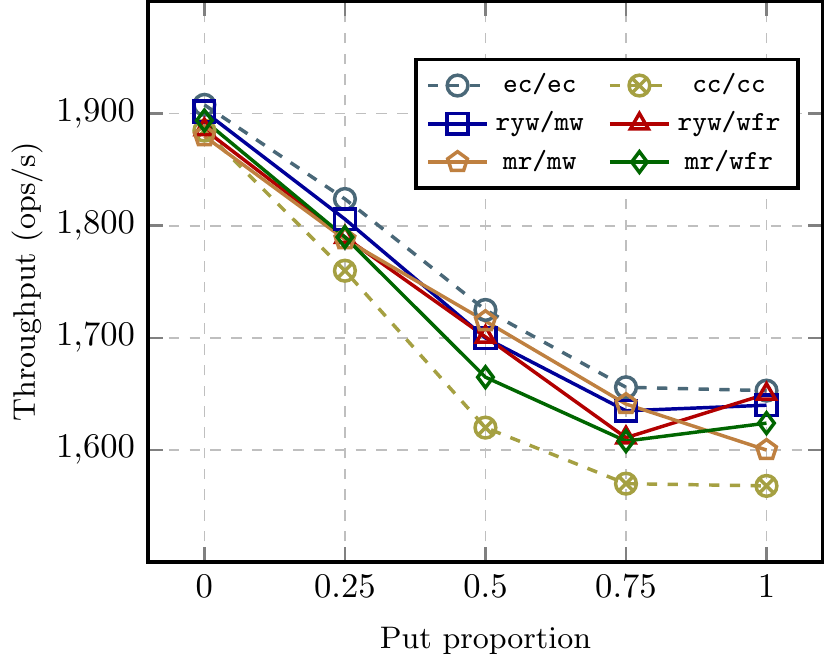}\label{fig:ratio-throughput}
    }
    \caption{The effect of read/write ratio on the latency and throughput with 100\% local traffic}\label{fig:ratio}
    \end{minipage}
    \begin{minipage}[h]{0.49\linewidth}
    \centering
    \subfigure[Latency]{
      \includegraphics[width = 0.46\textwidth]{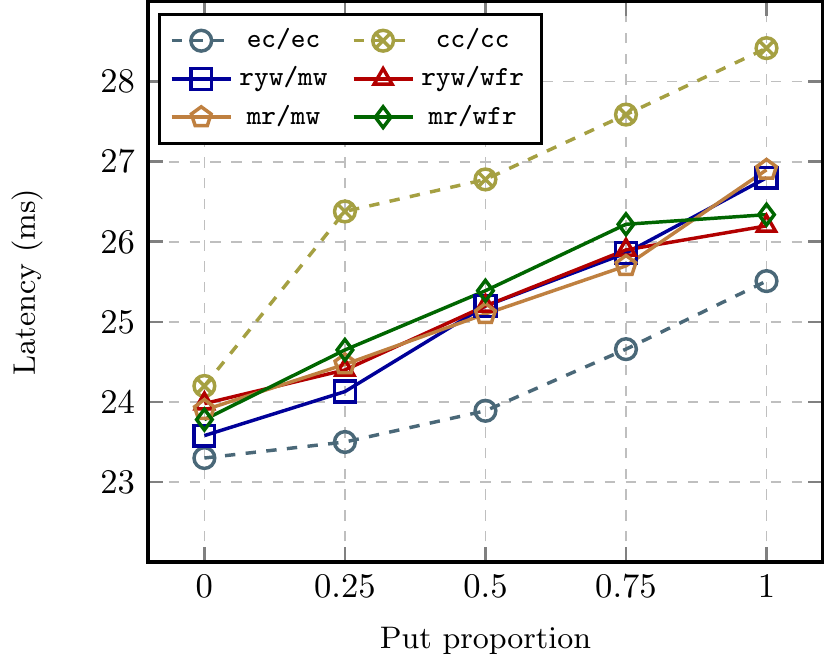}\label{fig:remote-ratio-latency}
    }
    \subfigure[Throughput]{
      \centering
      \includegraphics[width = 0.46\textwidth]{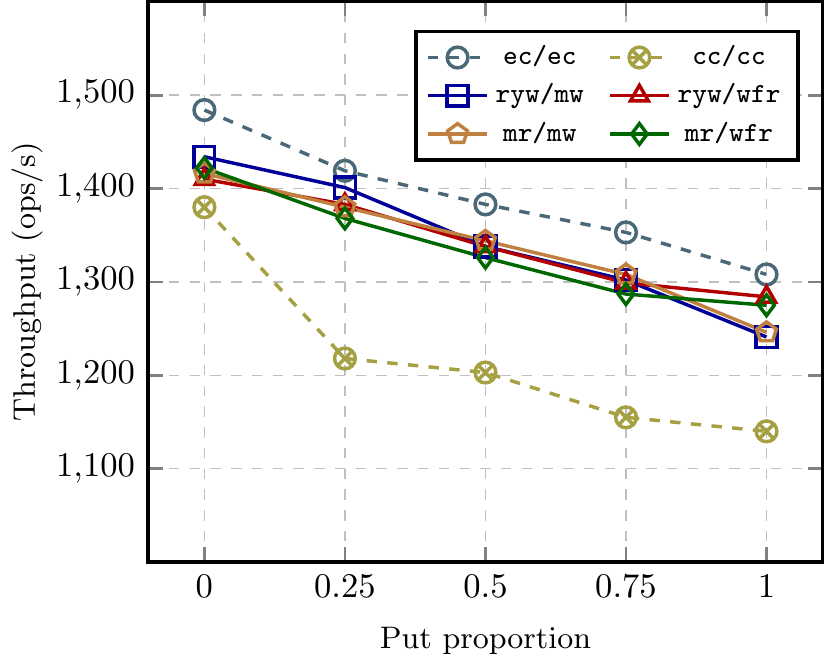}\label{fig:remote-ratio-throughput}
    }
    \caption{The effect of read/write ratio on the latency and throughput with 25\% remote traffic}\label{fig:remote-ratio}
    \end{minipage}
\end{figure*}

%% file: figs/load1.tex

\begin{figure*}[h]
  \centering
    \begin{minipage}[h]{0.49\linewidth}
    \centering
    \subfigure[Latency]{
      \includegraphics[width = 0.46\textwidth]{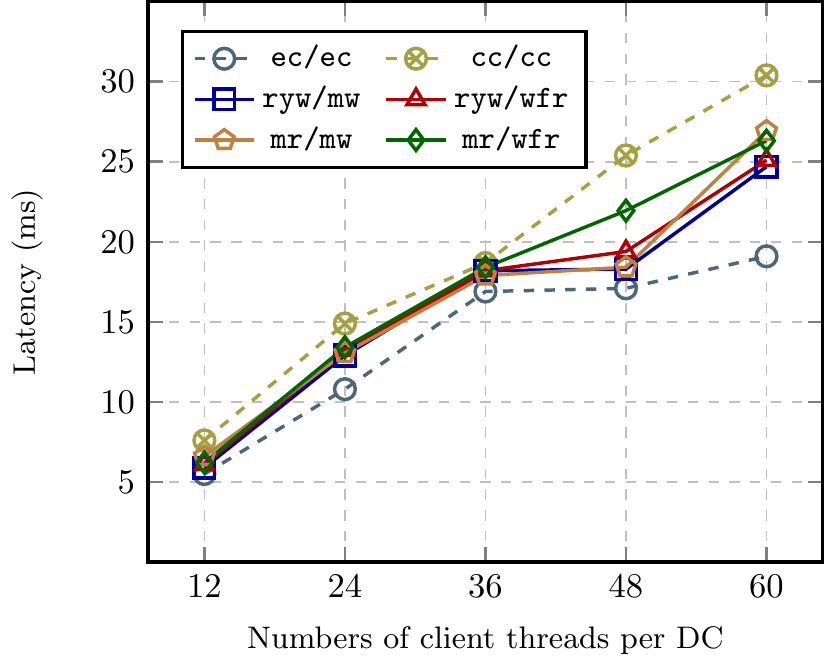}\label{fig:thread-latency}
    }
    \subfigure[Throughput]{
      \centering
      \includegraphics[width = 0.46\textwidth]{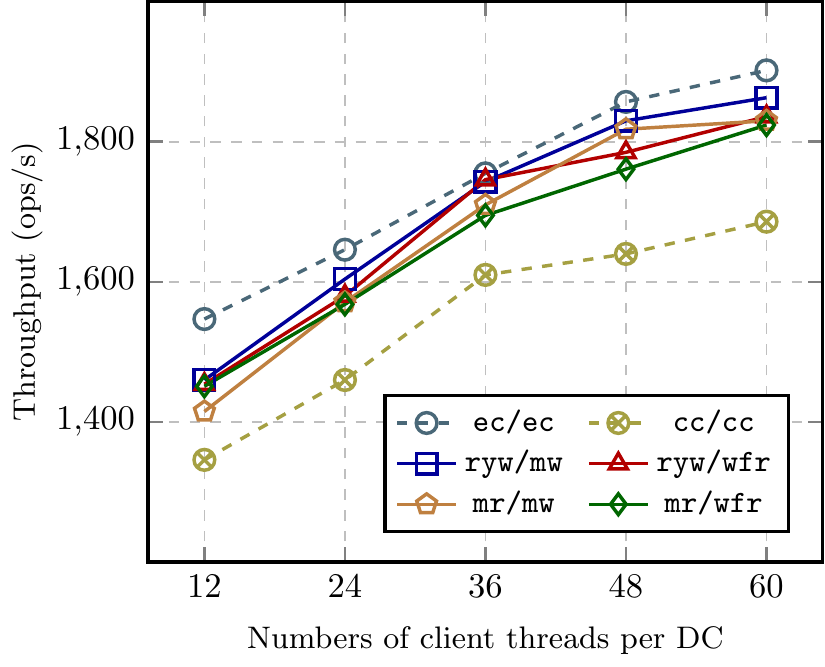}\label{fig:thread-throughput}
    }
    \caption{The effect of number of client threads on the latency and throughput with 100\% local traffic} \label{fig:load}
    \end{minipage}
    \begin{minipage}[h]{0.49\linewidth}
    \centering
    \subfigure[Latency]{
      \includegraphics[width = 0.46\textwidth]{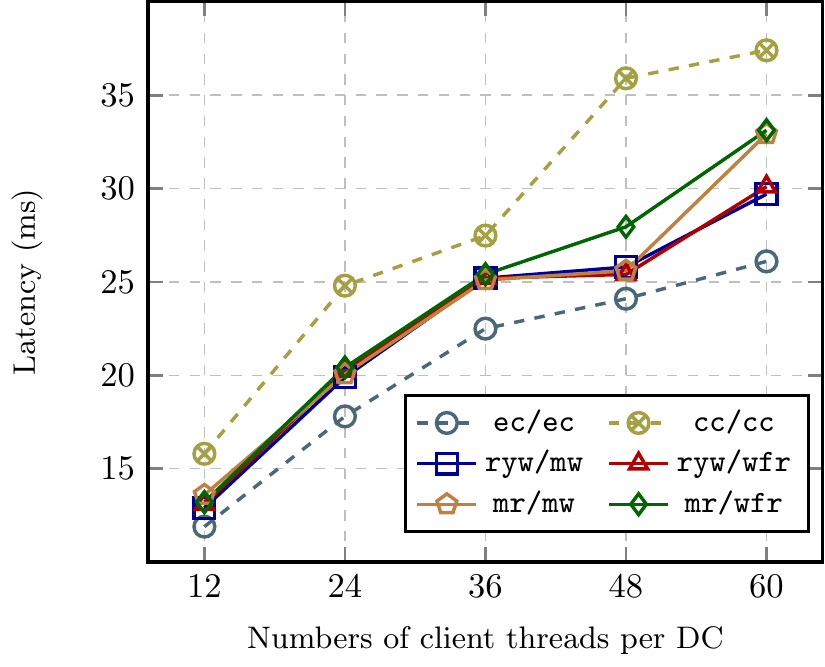}\label{fig:remote-thread-latency}
    }
    \subfigure[Throughput]{
      \centering
      \includegraphics[width = 0.46\textwidth]{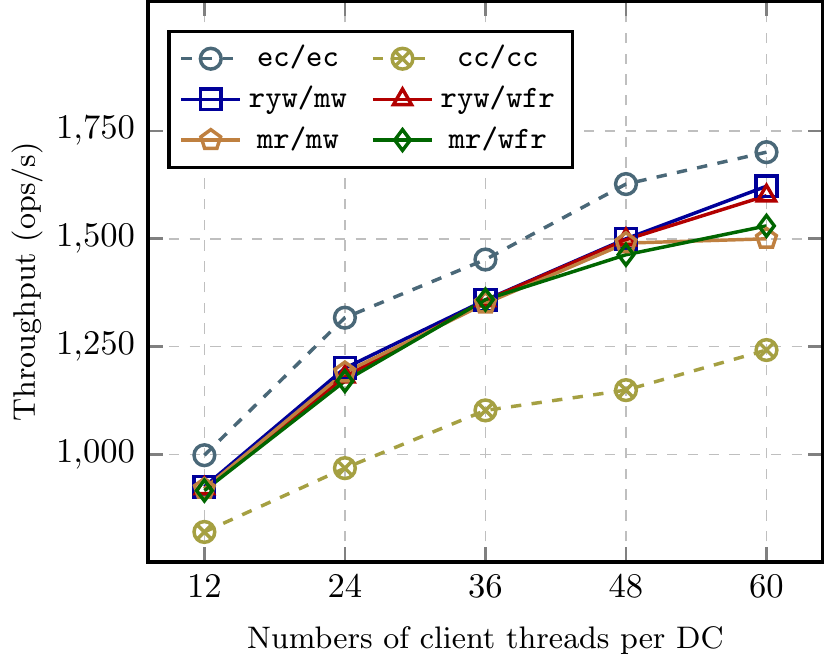}\label{fig:remote-thread-throughput}
    }
    \caption{The effect of number of client threads on the latency and throughput with 25\% remote traffic} \label{fig:remote-load}
    \end{minipage}
\end{figure*}

%% file: figs/partition.tex

\begin{figure}[h]
  \centering
  \subfigure[Latency]{
    \begin{minipage}[h]{0.46\linewidth}
      \centering
      \includegraphics[width = 1.0\textwidth]{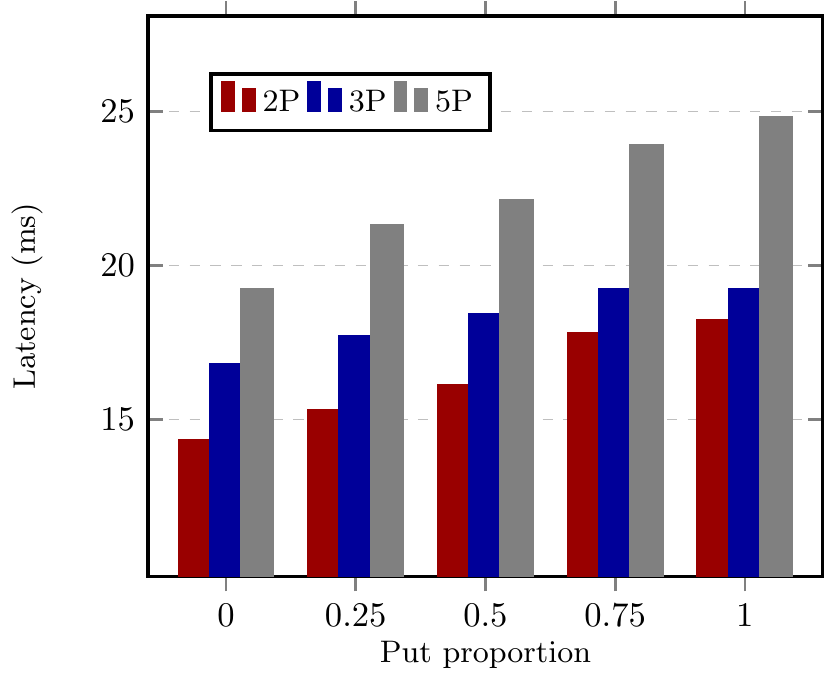}\label{fig:partition-latency}
    \end{minipage}
  }
  \subfigure[Throughput]{
    \begin{minipage}[h]{0.46\linewidth}
      \centering
      \includegraphics[width = 1.0\textwidth]{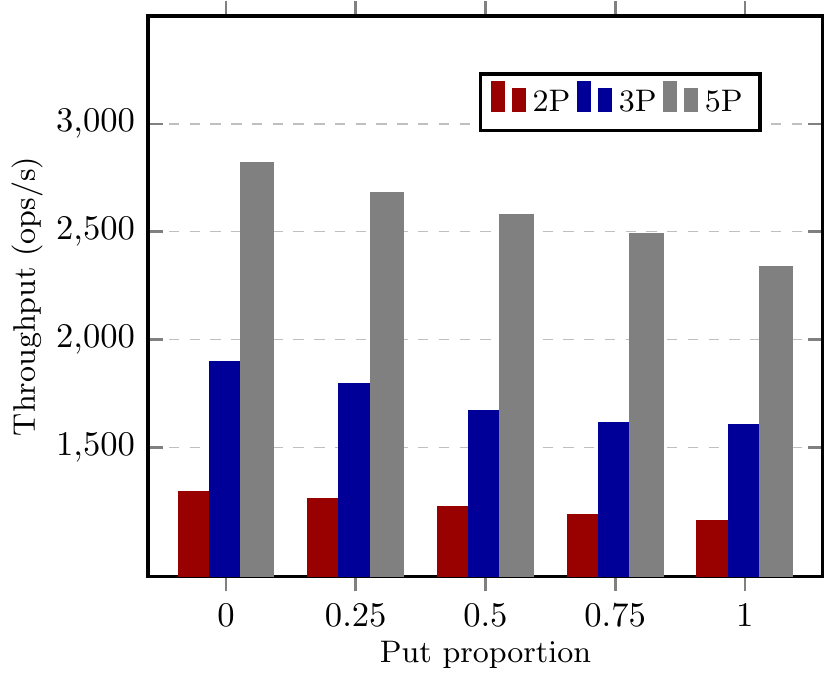}\label{fig:partition-throughput}
    \end{minipage}
  }
  \caption{The effect of deployment on the latency and throughput}
  \label{fig:partition}
\end{figure}

%% file: sections/related-work.tex

\section{Related Work} \label{section:related-work}

\emph{Specification Framework.} Burckhardt~\cite{PoEC:NOW2014} provides
the $(vis, ar)$ specification framework
for eventually consistent distributed data stores
based on the visibility and arbitration relations.
Bouajjani et al.~\cite{Multilevel:VMCAL2020} provides a framework for
specifying multilevel consistency. However, this framework is restricted to
only two consistency levels, and is specific to a concrete implementation of
data stores. In this work, we propose an implementation-independent
formalization of \tcc{} based on the four classic session guarantees.

\emph{Protocols of Tunable Consistency.}
Terry et al. present a protocol for ensuring session guarantees in~\cite{SESSION:PDIS1994},
which is adopted in Bayou~\cite{Bayou:SOSP1995}.
Bermbach et al.~\cite{Bermbach:ICCE2013} provide a middleware
to support \ryw{} and \mr{} on top of eventually consistent systems.
Roohitavaf et al.~\cite{Raft:DCN2019} present a protocol
for providing various session guarantees in distributed key-value stores.
To avoid slowdown cascades, they consider \emph{per-key} session guarantees.
Instead, we consider \emph{cross-key} session guarantees~\cite{SESSION:PDIS1994}.
Hence, we do not compare it in performance with our work in Section~\ref{section:exp}.

\emph{Systems Providing Tunable Consistency.}
Several popular data stores provide tunable consistency~\cite{DynamoDB:SOSP2007, Cassandra:SIGOPS2010, MongoDB:VLDB2019, CosmosDB:2018}.
Amazon DynamoDB~\cite{DynamoDB:SOSP2007} provides eventual consistency as the default
and stronger consistency with \texttt{ConsistentRead}.
Cassandra~\cite{Cassandra:SIGOPS2010} offers a number of fine-grained
read/write consistency levels such as \texttt{ANY}, \texttt{ONE}, \texttt{QUORUM}, \texttt{ALL}.
MongoDB~\cite{MongoDB:VLDB2019} exposes the \emph{readConcern} and \emph{writeConcern} parameters to clients,
which can be set per operation.
Azure Cosmos DB~\cite{CosmosDB:2018} offers five well-defined levels,
namely strong, bounded staleness, session, consistent prefix, and eventual.

\emph{Causally Consistent Systems.} There are plenty of research prototypes
and industrial deployments of causally consistent distributed data stores~\cite{COPS:SOSP2011,GentleRain:SOCC2014,Occult:NSDI2017,Contrarian:VLDB2018,OCC:TPDS2021}.
COPS~\cite{COPS:SOSP2011} explicitly tracks causal dependencies for updates.
GentleRain~\cite{GentleRain:SOCC2014} uses loosely synchronized physical clocks
to reduce the overhead of metadata and communication.
To reduce the delay of visibility of updates,
Occult~\cite{Occult:NSDI2017} and Contrarian~\cite{Contrarian:VLDB2018} use HLCs~\cite{HLC:ICPDS2014},
and POCC~\cite{OCC:TPDS2021} relies on both physical clocks and dependency vectors.

%% file: sections/conclusion.tex

\section{Conclusion and Future Work} \label{section:conclusion}

In this paper, we propose tunable causal consistency (TCC). TCC allows clients
to choose the desired session guarantee for each operation. We first propose a
formal specification of TCC. Then we design a TCC protocol supporting multiple
keys and develop a prototype distributed key-value store. To make TCC more
practically useful, it would be beneficial to explore how to automatically
choose the appropriate session guarantee for each operation in applications.

%% file: sections/appendix-protocol.tex

\section{Protocol} \label{appendix:protocol}

Algorithms~\ref{alg:client} and~\ref{alg:server} show the core of our protocol,
handling \pput{} and \get{} operations at client and server, respectively.
Algorithm~\ref{alg:clock} shows the pseudocode for clock management.

\input{algs/cure-client}
\input{algs/cure-server}
\input{algs/cure-clock}

%% file: algs/cure-client.tex

\begin{algorithm}
  \caption{Client operations at client $c$}
  \label{alg:client}
  \begin{algorithmic}[1]

    \Statex
    \Function{\readkey}{$k, l$}
      \State $p^{i}_{d} \gets \partition(k)$
      \If{$l = \ec$}
        \State $vc_r \gets O$
        \State $vc_w \gets O$
      \ElsIf{$l = \mr$}
        \State $vc_r \gets \hrvc$
        \State $vc_w \gets O$
      \ElsIf{$l = \ryw$}
        \State $vc_r \gets O$
        \State $vc_w \gets \hwvc$
      \EndIf

      \State $\langle v, vc, gsvc \rangle \gets \send \Call{\readrequest}{k, vc_r, vc_w}$ to $p^{i}_{d}$

      \For{$j = 1 \ldots D$}
  \State $\hrvc[j] \gets \max \set{\hrvc[j], gsvc[j]}$ \label{line:hrv-update}
	\State $\cvc_{c}^r[j] \gets \max \set{\cvc_{c}^r[j], vc[j]}$ \label{line:cvcr-update}
	\label{line:ts-get}
      \EndFor
      \State \Return $v$
    \EndFunction

    \Statex
    \Function{\updatekey}{$k, v, l$}
      \State $p^{i}_{d} \gets \Call{partition}{k}$

      \If{$l = \ec$}
        \State $dvc \gets O$
      \ElsIf{$l = \mw$}
        \State $dvc \gets \cvc_c^w$
      \ElsIf{$l = \wfr$}
        \State $dvc \gets \cvc_c^r$
      \EndIf

      \State $vc \gets \send \Call{\updaterequest}{k, v, dvc}$ to $p^{i}_{d}$
      \For{$j = 1 \ldots D$}
      \State $\cvc_{c}^w[j] \gets \max\{\cvc_{c}^w[j], vc[j]\}$
      \label{line:ts-put}
      \State $\hwvc[j]  \gets \max\{\hwvc[j], vc[j]\}$
      \label{line:hwv-update}
      \EndFor
      \State \Return \ok
    \EndFunction
  \end{algorithmic}
\end{algorithm}

%% file: algs/cure-server.tex

\begin{algorithm}
  \caption{Server operations at partition $p^{m}_{d}$}
  \label{alg:server}
  \begin{algorithmic}[1]

    \Statex
    \Function{\readrequest}{$k, vc_r, vc_w$}
      \State \wait \until $\forall i, \max \set{vc_r[i], vc_w[i]} \le \css^{m}_{d}[i]$
      \label{line:get-wait}
      \State $\langle v, vc \rangle \gets \Call{read}{k, \css^{m}_{d}, \store^{m}_{d}}$
    	\label{line:get-from}
            \Comment{\teal{\it returns the latest version for $k$ such that \purple{\emph{$vc \le \css^{m}_{d}$}}}}
      \State \Return $\langle v, vc, \css_m^d \rangle$
    \EndFunction

    \Statex
    \Function{\updaterequest}{$k, v, dvc$}
      \State \wait \until $dvc[d] < \clock^{m}_{d}$
    	\label{line:wait-clock}
      \hStatex
        \label{line:lc-put}
      \State $dvc[d] \gets \clock^{m}_{d}$
        \label{line:put-d-clock}
      \State $\pvc^{m}_{d}[d] \gets \clock^{m}_{d}$ \label{line:pvc-update}
      \State $\store^{m}_{d} \gets \store^{m}_{d} \cup \set{\langle k, v, dvc \rangle}$
    	\label{line:put}
      \State $\updates^{m}_{d} \gets \updates^{m}_{d} \cup \set{\langle k, v, dvc \rangle}$

      \hStatex
      \State \Return $dvc$
    \EndFunction

    \Statex
    \Function{\propagate}{\null}    \Comment{Run periodically}
      \State $\pvc^{m}_{d}[d] \gets \clock^{m}_{d}$

      \hStatex
      \If{$\updates^{m}_{d} \neq \emptyset$}
        \For{$\langle k, v, vc \rangle \in \updates^{m}_{d}$}  \Comment{in $vc[d]$ order}
          \For{$i = 1 \ldots D, i \neq d$}
            \State \send \Call{\replicate}{$d, k, v, vc$} to $p^{m}_{i}$
          \EndFor
        \EndFor
        \State $\updates^{m}_{d} \gets \emptyset$
      \Else
        \For{$i = 1 \ldots D, i \neq d$}
          \State \send \Call{\heartbeat}{$d, \pvc^{m}_{d}[d]$} to $p^{m}_{i}$
        \EndFor
      \EndIf
    \EndFunction

    \Statex
    \Function{\heartbeat}{$i, ts$}
      \State $\pvc^{m}_{d}[i] \gets ts$
    \EndFunction

    \Statex
    \Function{\replicate}{$i, k, v, vc$}
      \State $\pvc^{m}_{d}[i] \gets vc[i]$
      \State $\store^{m}_{d} \gets \store^{m}_{d} \cup \set{(k, v, vc)}$
    \EndFunction
  \end{algorithmic}
\end{algorithm}

%% file: algs/cure-clock.tex

\begin{algorithm}
  \caption{Clock management at $p^{m}_{d}$}
  \label{alg:clock}
  \begin{algorithmic}[1]
    \Function{\bcast}{\null}    \Comment{Run periodically}
      \For{$i = 1 \ldots N, i \neq m$}
        \State \send \Call{\updatecss}{$m, \pvc^{m}_{d}$} to $p^{i}_{d}$
      \EndFor
    \EndFunction

    \Statex
    \Function{\updatecss}{$i, pvc$}
      \State $\pmc^{m}_{d}[i] \gets pvc$
      \For{$j = 1 \ldots D, j \neq d$}
        \label{line:updatecss-no-d}
    	\State $\css^{m}_{d}[j] \gets \min\limits_{i = 1 \ldots N} \pmc^{m}_{d}[i][j]$
    	  \label{line:compute-css}
      \EndFor
    \EndFunction
  \end{algorithmic}
\end{algorithm}

%% file: sections/appendix-proof.tex

\section{Correctness Proof}  \label{appendix:proof}

In this section, we prove the correctness of the protocol above.
We first define three sets of \pput{} events.
We define $\cwrites(c,t)$ to be
the set of all \pput{} events
done by client $c$ at time $t$.
Similarly, we define $\creads(c,t)$ to be
the set of all \pput{} events that
have written a value for some key
read by client $c$ at time $t$.
We also define $\swrites(s,t)$ to be
the set of all stable \pput{} events
known to $s$ in the local datacenter
when $c$'s request is applied at $s$ at time $t$.
Now we use the definitions above to
describe four session guarantees
defined in Section~\ref{section:tcc}.

Read your writes (\ryw): Let $g$ be a read by client $c$ at time $t$.
Event $g$ satisfies \ryw{}
if any $p\in \cwrites(c,t)$ is included in $\swrites(s,t')$
, where time $t'$ is when $g$ is applied at partition $s$.

Monotonic read (\mr): Let $g$ be a read by client $c$ at time $t$.
Event $g$ satisfies \mr{}
if any $p\in \creads(c,t)$ is included in $\swrites(s,t')$
, where time $t'$ is when $g$ is applied at partition $s$.

Monotonic write (\mw): Let $p$ be a write by client $c$ at time $t$.
Partition $s$ applies it at time $t'$.
Event $p$ satisfies \mw{}
if $p$ is included in $\swrites(s,t')$,
there is no \pput{} event $e$ ($e \neq p$) such that
$\cwrites(c,t)$ includes $e$
and $vc_p$ is smaller than $vc_e$.

Writes follow read (\wfr): Let $p$ be a write by client $c$ at time $t$.
Partition $s$ applies it at time $t'$.
Event $p$ satisfies \wfr{}
if $p$ is included in $\swrites(s,t')$,
there is no \pput{} event $e$ ($e \neq p$) such that
$\creads(c,t)$ includes $e$
and $vc_p$ is smaller than $vc_e$.




Now, we prove the correctness of the protocol
for session guarantees constraining \get{} events.

\begin{theorem}[$\mr$]  \label{thm:mr}
  Any \get{} event reading a key by client $c$
  at partition $p_{d}^{m}$ with $\hrvc$ satisfies $\mr$.
\end{theorem}

\begin{proof}
  Let $g$ be a \get{} event by client $c$
  reading the value of key $k$ at time $t$
  at partition $p_{d}^{m}$ and returns at time $t'$.
  Since $g$ returns at time $t'$, we have $\css_{d}^{m} \ge \hrvc$ at time $t'$
  according to Line~\ref{line:get-wait} of Algorithm~\ref{alg:server}.
  Suppose by a contradiction that
  $g$ does not satisfy $\mr$.
  This means that there exists a \pput{} event $e$ such that
  $e \in \creads(c,t)$ and $e \notin \swrites(p_{d}^{m},t')$.
  Since $e \notin \swrites(p_{d}^{m},t')$, there is some datacenter $i$ such that
  $\css_{d}^{m}[i] < vc_e[i]$ at time $t'$ according to Line~\ref{line:get-wait} of Algorithm~\ref{alg:server}.
  Since $e \in \creads(c,t)$, $\hrvc \ge vc_e$
  according to Line~\ref{line:hrv-update} of Algorithm~\ref{alg:client}
  and Line~\ref{line:get-from} of Algorithm~\ref{alg:server}.
  Thus, $\css_{d}^{m} \ge vc_e$ at time $t'$.
  Contradiction.
\end{proof}

We can prove the following theorem in a similar way.

\begin{theorem}[$\ryw$]  \label{thm:ryw}
  Any \get{} event reading a key by client $c$
  at partition $p_{d}^{m}$ with $\hwvc$ satisfies $\ryw$.
\end{theorem}


Now, we prove the correctness of the protocol
for session guarantees constraining \pput{} events.

\begin{theorem}[$\mw$]  \label{thm:mw}
  Any \pput{} event writing a key by client $c$
  at partition $p_{d}^{m}$ with $\cvc_c^w$ satisfies $\mw$.
\end{theorem}

\begin{proof}
  Let $p$ be a \pput{} event by client $c$
  writing to key $k$ at time $t$
  at partition $p_{d}^{m}$ and returns at $t'$.
  Suppose by a contradiction that
  $p$ does not satisfy $\mw$.
  This means that there exists
  a \pput{} event $e$ ($e \neq p$) such that
  $e \in \cwrites(c,t)$ and $vc_p < vc_e$.
  We have $\cvc_c^{w} \ge vc_e$ at time $t$
  according to Line~\ref{line:ts-put} of Algorithm~\ref{alg:client}.
  By Line~\ref{line:wait-clock} and Line~\ref{line:put-d-clock} of Algorithm~\ref{alg:server},
  we have $vc_p > \cvc_c^w$ at time $t'$. Thus, $vc_p > vc_e$.
  Contradiction.
\end{proof}

We can prove the following theorem in a similar way.

\begin{theorem}[$\wfr$]  \label{thm:wfr}
  Any \pput{} event writing a key by client $c$
  at partition $p_{d}^{m}$ with $\cvc_c^r$ satisfies $\wfr$.
\end{theorem}
